\documentclass[12pt,final]{l4dc2023}

\title[Template-Based Piecewise Affine Regression]{Template-Based Piecewise Affine Regression}

\usepackage{times} 

\usepackage{array}
\usepackage{booktabs}
\usepackage{paralist}
\usepackage{tikz}
\usetikzlibrary{positioning,arrows,patterns,calc}
\usepackage{pgfplots}
\pgfplotsset{compat=newest}

\author{%
 \Name{Guillaume O. Berger}  \and \Name{Sriram Sankaranarayanan} \Email{first.lastname@colorado.edu}\\
 \addr University of Colorado Boulder, USA.
}

\newcommand{\calD}{\mathcal{D}}

\newcommand{\calH}{\mathcal{H}}
\newcommand{\calI}{\mathcal{I}}

\newcommand{\calP}{\mathcal{P}}

\newcommand{\calS}{\mathcal{S}}

\newcommand{\calU}{\mathcal{U}}
\newcommand{\calV}{\mathcal{V}}

\newcommand{\calX}{\mathcal{X}}

\newcommand{\chat}{\hat{c}}

\newcommand{\qhat}{\hat{q}}

\newcommand{\epsilonhat}{\hat\epsilon}

\newcommand{\xbar}{\bar{x}}

\renewcommand{\Re}{\mathbb{R}}

\newcommand{\Ne}{\mathbb{N}}

\newcommand{\size}{\mathrm{size}}

\newcommand{\FindSubsets}{\textsc{FindSubsets}}

\newcommand{\cbreak}{\textsc{break}}
\newcommand{\ccontinue}{\textsc{continue}}

\newcommand{\Uid}{U_{\mathrm{id}}}
\newcommand{\infeas}{\mathrm{ic}}

\newtheorem{problem}[theorem]{Problem}

\newenvironment{proofwithname}[1][Name]%
{%
 \par\noindent{\bfseries\upshape #1\ }%
}%
{\jmlrQED}

\begin{document}

\maketitle

\begin{abstract}%
We investigate the problem of fitting piecewise affine functions (PWA) to data.
Our algorithm divides the input domain into finitely many polyhedral regions
whose shapes are specified using a user-defined template such that the data
points in each region are fit by an affine function within a desired error
bound. We first prove that this problem is NP-hard. Next, we present a top-down
algorithm that considers subsets of the overall data set in a systematic manner,
trying to fit an affine function for each subset using linear regression. If
regression fails on a subset, we extract a minimal set of points that led to a
failure in order to split the original index set into smaller subsets. Using a
combination of this top-down scheme and a set covering algorithm, we derive an
overall approach that is optimal in terms of the number of pieces of the
resulting PWA model. We demonstrate our approach on two numerical examples that
include PWA approximations of a widely used nonlinear insulin--glucose
regulation model and a double inverted pendulum with soft contacts.
\end{abstract}





\section{Introduction}\label{section:introduction}
Piecewise affine (PWA) regression models a given set of data points consisting of input--output pairs $\{(x_k,y_k)\}_{k=1}^K$ by splitting the input domain into finitely many polyhedral regions $H_1, \ldots, H_q$ and associating each region $H_i$ with an affine function $f_i(x)=A_i x+b_i$.
In this paper, we seek a PWA model that fits the given data while respecting a user-provided error bound $\epsilon$ and minimizing the number of regions. This problem has numerous applications including the identification of hybrid systems with state-based switching and simplifying nonlinear models using PWA approximations.


Existing PWA regression approaches usually do not restrict how the input domain
is split. For instance, an approach that simply specifies that the input space
is covered by polyhedral sets leads to high computational complexity for the
regression algorithm~\citep{lauer2019hybrid}. In this paper, we restrict the
possible shape of the polyhedral regions by requiring that each region $H_i$ is
described by a vector inequality $p(x)\leq c_i$, wherein $p$ is a fixed,
user-defined, vector-valued function, called the \emph{template}, while the
regions are obtained by varying the offset vector $c_i$. The resulting
problem, called template-based PWA regression, allows us to split the input
region into pre-specified shapes such as rectangles, using a suitable template.
Like the classical PWA regression problem~\citep{lauer2019hybrid}, we show that
template-based PWA regression is NP-hard in the dimension of the
input space and the size of the template, but polynomial in the size of the data
set (Section \ref{section:complexity}).

Next, we provide an algorithm for optimal template-based PWA regression (i.e., with minimal number of regions) (Section \ref{section:algorithm}).
The main idea is to examine various subsets of the input data in order to discover maximal subsets that are \emph{compatible}: wherein compatibility of a set of data points simply means that there is an affine function that fits all the points within the desired error tolerance.
Thus, our approach starts to examine subsets of the data starting from the entire data to begin with.
If a given subset is not compatible, we exploit the optimization formulation of the affine regression problem to extract a minimal subset of points that is itself incompatible.
The key observation is that the original set can now be broken up into smaller subsets which can themselves be examined for compatibility.
We show that by integrating this process with a minimal set cover algorithm, we can extract a partition with the smallest size that in turn leads to the desired PWA model. 

We apply our framework on two practical problems: the approximation of a nonlinear system, namely the insulin--glucose regulation process \citep{dalla2007meal}, with affine functions with rectangular domains (Subsection \ref{subsection:exa-glucose}), and the identification of a hybrid linear system consisting in an inverted double pendulum with soft contacts on the joints (Subsection \ref{subsection:exa-pendulum}).
For both applications, we show that template-based PWA regression is favorable compared to classical PWA regression both in terms of computation time and our ability to formulate models from the results.

\subsection{Related work}
Piecewise affine systems and hybrid linear systems appear naturally in a wide range of applications \citep{jungers2009thejoint}, or as approximations of more complex systems \citep{breiman1993hinging}.
Therefore, the problems of switched affine (SA) and piecewise affine (PWA) regression have received a lot of attention in the literature; see, e.g., \citet{paoletti2007identification,lauer2019hybrid} for surveys.
Both problems are known to be NP-hard \citep{lauer2019hybrid}.
The problem of SA regression can be formulated as a Mixed-Integer Program and solved using MIP solvers, but the complexity is exponential in the number of data points \citep{paoletti2007identification}.
\citet{vidal2003analgebraic} propose an efficient algebraic approach to solve the problem, but it is restricted to noiseless data.
Heuristics to solve the problem in the general case include greedy algorithms \citep{bemporad2005aboundederror}, continuous relaxations of the MIP \citep{munz2005continuous}, block–coordinate descent (similar to $k$-mean regression) algorithms \citep{bradley2000kplane,lauer2013estimating} and refinement of the algebraic approach using sum-of-squares relaxations \citep{ozay2009robust}; however, these methods offer no guarantees of finding an (optimal) solution to the problem.
As for PWA regression, classical approaches include clustering-based methods \citep{ferraritrecate2005aclustering}, data classification followed by geometric clustering \citep{nakada2005identification} and block–coordinate descent algorithms \citep{bemporad2022apiecewise}; however, these methods are not guaranteed to find a (minimal) piecewise affine model.

Piecewise affine systems with constraints on the domain appear naturally in several applications including biology \citep{porreca2009identification} and mechanical systems with contact forces \citep{aydinoglu2020contactaware}, or as approximations of nonlinear systems \citep{smarra2020datadriven}.
Techniques for PWA regression with rectangular domains have been proposed in \citet{munz2002identification,smarra2020datadriven}; however, these approaches impose further restrictions on the arrangement of the domains of the functions (e.g., forming a grid) and they are not guaranteed to find a solution with a minimal number of pieces.
In the one-dimensional case (e.g., time series), an exact efficient algorithm for optimal PWA regression was proposed by \citet{ozay2012asparsification}, but the approach does not extend to higher dimension.
As for the application involving mechanical systems with contact forces (presented in Subsection \ref{subsection:exa-pendulum}), a recent work by \citet{jin2022learning} proposes a heuristic based on minimizing a loss function to learn \emph{linear complementary systems}.

\subsection{Approach at a glance}
\begin{figure}[t]
\centering
\tikzset{
  IndexSet/.style={draw},
  InfeasCert/.style={rounded corners,fill=red!50},
  Valid/.style={rounded corners,fill=green!50}
}
\pgfplotsset{
  every axis/.style={
    width=0.54\textwidth,
    height=4cm,
    xtick={-1,-0.5,0,0.5,1},
    ytick={-1,-0.5,0,0.5,1},
    domain=-1:1,
    clip marker paths=true,
    declare function={
      nonlin(\x) = atan(10*\x)*exp(-abs(\x))/180*pi;
      lini(\x) = -0.4565738459677249*\x - 1.034561584207168;
      linii(\x) = 4.532283517444388*\x;
      liniii(\x) = -0.4565738459677249*\x + 1.034561584207168;
    }
  }
}
\subfigure[]{\begin{tikzpicture}[scale=0.55]
  \begin{axis}[
    at={(0,0)},
    ylabel={$y$},
    legend pos=north west,
    legend style={draw=none},
    xticklabels={},
  ]
    \addplot[dashed,line width=2pt,samples=100,gray] {nonlin(x)};
    \addlegendentry{function}
    \addplot[only marks,samples=11,mark size=3pt] {nonlin(x)};
    \addlegendentry{data points}
    \addplot[only marks,samples at={-1},mark size=0pt] {nonlin(x)} node[above] {1};
    \addplot[only marks,samples at={-0.8},mark size=0pt] {nonlin(x)} node[above] {2};
    \addplot[only marks,samples at={-0.6},mark size=0pt] {nonlin(x)} node[above] {3};
    \addplot[only marks,samples at={-0.4},mark size=0pt] {nonlin(x)} node[above] {4};
    \addplot[only marks,samples at={-0.2},mark size=0pt] {nonlin(x)} node[above] {5};
    \addplot[only marks,samples at={0},mark size=0pt] {nonlin(x)} node[right] {6};
    \addplot[only marks,samples at={0.2},mark size=0pt] {nonlin(x)} node[below] {7};
    \addplot[only marks,samples at={0.4},mark size=0pt] {nonlin(x)} node[below] {8};
    \addplot[only marks,samples at={0.6},mark size=0pt] {nonlin(x)} node[below] {9};
    \addplot[only marks,samples at={0.8},mark size=0pt] {nonlin(x)} node[below] {10};
    \addplot[only marks,samples at={1},mark size=0pt] {nonlin(x)} node[below] {11};
    \node[draw,fill=white,anchor=south east,outer sep=0pt] at (rel axis cs:1,0) {\bf I};
  \end{axis}
  \begin{axis}[
    at={(0.46\textwidth,0)},
    legend pos=north west,
    legend style={draw=none},
    xticklabels={},
    yticklabels={},
  ]
    \addplot[dashed,line width=2pt,samples=100,gray,forget plot] {nonlin(x)};
    \addplot[only marks,samples=11,mark size=3pt,color=blue] {nonlin(x)};
    \addlegendentry{index set}
    \addplot[
      only marks,samples at={-0.4,-0.2,0},mark size=4pt,
      mark options={draw=red!90,fill opacity=0,line width=1.5pt},
    ] {nonlin(x)};
    \addlegendentry{\,infeas.\ cert.}
    \node[draw,fill=white,anchor=south east,outer sep=0pt] at (rel axis cs:1,0) {\bf II};
  \end{axis}
  \begin{axis}[
    at={(0,-2.8cm)},
    ylabel={$y$},
    legend pos=north west,
    legend style={draw=none},
    xticklabels={},
  ]
    \addplot[dashed,line width=2pt,samples=100,gray,forget plot] {nonlin(x)};
    \addplot[only marks,samples=11,mark size=3pt,forget plot] {nonlin(x)};
    \addplot[only marks,domain=-1:-0.2,samples=5,mark size=3pt,color=blue] {nonlin(x)};
    \addlegendentry{index set}
    \addplot[domain=-1:-0.2,samples=5,line width=2pt,color=green] {lini(x)};
    \addlegendentry{line 1}
    \node[draw,fill=white,anchor=south east,outer sep=0pt] at (rel axis cs:1,0) {\bf III};
  \end{axis}
  \begin{axis}[
    at={(0.46\textwidth,-2.8cm)},
    legend pos=north west,
    legend style={draw=none},
    xticklabels={},
    yticklabels={},
  ]
    \addplot[dashed,line width=2pt,samples=100,gray,forget plot] {nonlin(x)};
    \addplot[only marks,samples=11,mark size=3pt,forget plot] {nonlin(x)};
    \addplot[only marks,domain=-0.2:1,samples=7,mark size=3pt,color=blue] {nonlin(x)};
    \addlegendentry{index set}
    \addplot[
      only marks,samples at={0.4,0.2,0},mark size=4pt,
      mark options={draw=red!90,fill opacity=0,line width=1.5pt},
    ] {nonlin(x)};
    \addlegendentry{\,infeas.\ cert.}
    \node[draw,fill=white,anchor=south east,outer sep=0pt] at (rel axis cs:1,0) {\bf IV};
  \end{axis}
  \begin{axis}[
    at={(0,-5.6cm)},
    xlabel={$x$},
    ylabel={$y$},
    legend pos=north west,
    legend style={draw=none},
  ]
    \addplot[dashed,line width=2pt,samples=100,gray,forget plot] {nonlin(x)};
    \addplot[only marks,samples=11,mark size=3pt,forget plot] {nonlin(x)};
    \addplot[only marks,domain=-0.2:0.2,samples=3,mark size=3pt,color=blue] {nonlin(x)};
    \addlegendentry{index set}
    \addplot[domain=-0.2:0.2,samples=5,line width=2pt,color=green] {linii(x)};
    \addlegendentry{line 2}
    \node[draw,fill=white,anchor=south east,outer sep=0pt] at (rel axis cs:1,0) {\bf V};
  \end{axis}
  \begin{axis}[
    at={(0.46\textwidth,-5.6cm)},
    xlabel={$x$},
    legend pos=north west,
    legend style={draw=none},
    yticklabels={},
  ]
    \addplot[dashed,line width=2pt,samples=100,gray,forget plot] {nonlin(x)};
    \addplot[only marks,samples=11,mark size=3pt,forget plot] {nonlin(x)};
    \addplot[only marks,domain=0.2:1,samples=5,mark size=3pt,color=blue] {nonlin(x)};
    \addlegendentry{index set}
    \addplot[domain=0.2:1,samples=5,line width=2pt,color=green] {liniii(x)};
    \addlegendentry{line 3}
    \node[draw,fill=white,anchor=south east,outer sep=0pt] at (rel axis cs:1,0) {\bf VI};
  \end{axis}
\end{tikzpicture}}
\hfill
\subfigure[]{\scalebox{0.6}{\begin{tikzpicture}
  \node[IndexSet] (a) {$I_0=\{1,\ldots,11\}$};
  \node[InfeasCert,right=5pt of a] (a-cert) {$C_0=\{4,5,6\}$};
  \node[above=0pt of a]{\fbox{\textbf{II}}};
  \node[IndexSet,below left=1.5cm and -3mm of a] (b) {$I_1=\{1,\ldots,5\}$};
  \node[Valid,below=5pt of b] (b-line) {line 1};
  \node[above=0pt of b]{\fbox{\textbf{III}}};
  \draw[-latex] (a) -- (b) node[midway,right] {$\neg6$};
  \node[IndexSet,below right=1.5cm and -3mm of a] (c) {$I_2=\{5,\ldots,11\}$};
  \node[InfeasCert,right=5pt of c] (c-cert) {$C_2=\{6,7,8\}$};
  \node[above=0pt of c]{\fbox{\textbf{IV}}};
  \draw[-latex] (a) -- (c) node[midway,right] {$\neg4$};
  \node[IndexSet,below left=1.5cm and -3mm of c] (d) {$I_3=\{5,6,7\}$};
  \node[Valid,below=5pt of d] (d-line) {line 2};
  \node[above=0pt of d]{\fbox{\textbf{V}}};
  \draw[-latex] (c) -- (d) node[midway,right] {$\neg8$};
  \node[IndexSet,below right=1.5cm and -3mm of c] (e) {$I_4=\{7,\ldots,11\}$};
  \node[Valid,below=5pt of e] (e-line) {line 3};
  \node[above=0pt of e]{\fbox{\textbf{VI}}};
  \draw[-latex] (c) -- (e) node[midway,right] {$\neg6$};
\end{tikzpicture}}}
\caption{(a) Illustration of our algorithm on a simple data set with $11$ data points $(x_k, y_k) \in \Re \times \Re$ and (b) the index sets
explored by our algorithm. }
\label{fig:approach-glance}
\end{figure}
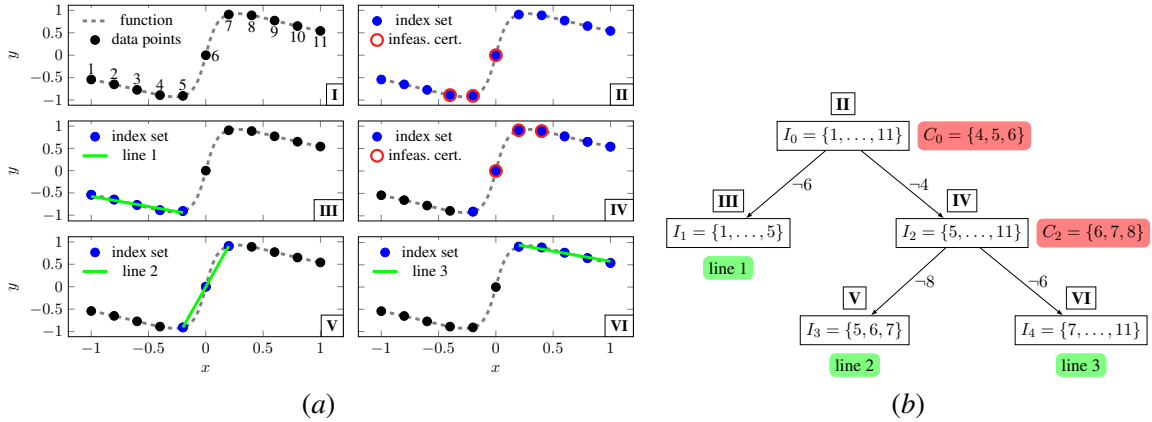

Figure~\ref{fig:approach-glance} (plot label ``I'') shows the working of our algorithm on a simple data set with $K=11$ points $(x_k, y_k)\in\Re\times\Re$.
We seek a piecewise affine (PWA) function that fits the data within the error tolerance $\epsilon = 0.1$ with the smallest number of affine functions defined on intervals.
At the very first step (label ``II''), the approach tries to fit a single straight line through all the $11$ points.
This corresponds to the \emph{index set} $I_0=\{1,\ldots,11\}$ where the indices correspond to points as shown in the plot ``I''.
However, no such line can fit the points for the given $\epsilon$.
Our approach generates an \emph{infeasibility certificate} that identifies the indices $C_0=\{4,5,6\}$ as a cause of this infeasibility (see plot ``II'').
In other words, we cannot have all three points in $C_0$ be part of the same piece of the PWA function we seek.
Therefore, our approach now splits $I_0$ into two subsets $I_1=\{1, \ldots, 5\}$ and $I_2=\{5, \ldots, 11\}$.
These  two sets are maximal intervals with respect to set inclusion and do not contain $C_0$.
The set $I_1$ can be fit by a single straight line with tolerance $\epsilon$ (see plot ``III'').
However, considering $I_2$, we notice once again that a single straight line cannot be fit (see plot ``IV'').
We identify the set $C_2=\{6,7,8\}$ as an infeasibility certificate and our algorithm splits $I_2$ into maximal subsets $I_3=\{5,6, 7\}$ and $I_4=\{7,\ldots,11\}$.
Each of these subsets can be fit by a straight line (see plots ``V'' and ``VI'').
Thus, our approach finishes by discovering three pieces that cover all the points $\{1,\ldots,11\}$.
Note that although the data point indexed by $5$ is part of two pieces, we can resolve this ``tie'' in an arbitrary manner by assigning $5$ to the first piece and removing it from the second; the same holds for the data point indexed by $7$.

\bigskip

Due to space limitation, the proofs of several results presented in the paper can be found in the extended version of the paper, available on arXiv.

\section{Problem Statement}\label{section:problem-statement}

Given $K\in\Ne_{>0}$ observation data points $\{(x_k,y_k)\}_{k=1}^K\subseteq\Re^d\times\Re^e$ (see Figures \ref{fig:illustration}\emph{(a,c)}), we wish to find a piecewise affine (PWA) function that fits the data within some given error tolerance $\epsilon\geq0$.
Formally, a PWA function over a domain $D\subseteq\Re^d$ is defined by covering the domain with $q$ regions $H_1, \ldots, H_q$ and associating an affine function $f_i(x) = A_i x + b_i$ with each $H_i$:
\[
f(x) =  A_1 x + b_1\ \text{if}\ x \in H_1,\ \dots,\ A_i x + b_i\ \text{if}\ x \in H_i,\ \dots,\ A_q x + b_q \ \text{if}\ x \in H_q. 
\]
If $H_i\cap H_j\neq\emptyset$ for $i\neq j$, then $f$ is no longer a function.
However, in such a case, we may ``break the tie'' by defining $f(x) = f_i(x)$ wherein $i = \min\,\{ j\ |\ x \in H_j\}$.

\begin{figure}
\centering
\subfigure[Data set $\calD_1$]{\includegraphics[width=0.245\textwidth]{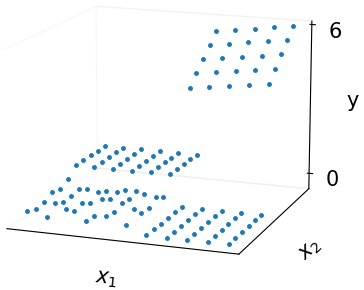}}
\hfill
\subfigure[$\calD_1$ fit: $\epsilon=0.375$]{\includegraphics[width=0.245\textwidth]{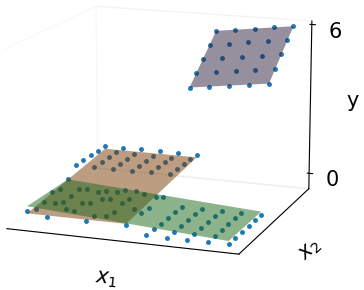}}
\hfill
\subfigure[Data set $\calD_2$]{\includegraphics[width=0.245\textwidth]{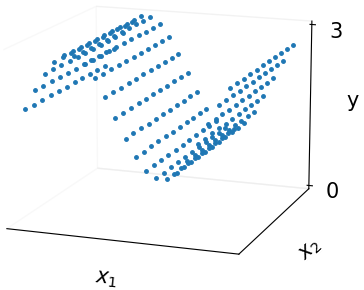}}
\hfill
\subfigure[$\calD_2$ fit:  $\epsilon=0.15$]{\includegraphics[width=0.245\textwidth]{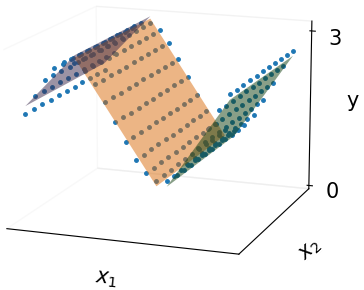}}
\caption{Template-based piecewise affine (TPWA) regression.
\emph{(a)}, \emph{(c)}: Data points $(x_k,y_k)\in\Re^2\times\Re$.
\emph{(b)}, \emph{(d)}: TPWA fit with rectangular domains and error tolerance $\epsilon$.}
\label{fig:illustration}
\vspace*{-0.5cm}
\end{figure}



\begin{problem}[PWA regression]\label{prob:OPWAR}%
Given data $\{(x_k,y_k)\}_{k=1}^K$ and an error bound $\epsilon\geq0$, find $q$ regions $H_i\subseteq\Re^d$ and affine functions $f_i(x) = A_i x + b_i$ such that
\begin{equation}\label{equat:error}
\forall\,k,\:\exists\,i : x_k\in H_i  \quad\text{and}\quad \forall\,k,\:\forall\,i,\: x_k\in H_i  \Rightarrow \lVert y_k-f_i(x_k)\rVert_\infty \leq \epsilon.
\end{equation}
\end{problem}

Furthermore, we restrict the domain $H_i$ of each affine piece by specifying a \emph{template}, which can be any function $p:\Re^d\to\Re^h$.
Given a template $p$ and a vector $c\in\Re^h$, we define the set $H(c)$ as
\begin{equation}\label{equat:template-region}
H(c) = \{ x \in \Re^d : p(x) \leq c \},
\end{equation}
wherein $\leq$ is elementwise and $c\in\Re^h$ parameterizes the set $H(c)$.
We let $\calH=\{H(c):c\in\Re^h\}$ denote the set of all regions in $\Re^d$ described by the template $p$.

Fixing a template \emph{a priori} controls the complexity of the domains, and thus of the overall PWA function.
The \emph{rectangular} template $p(x)=[x;-x]$ defines regions $H(c)$ that form boxes in $\Re^d$.
Similarly, allowing pairwise differences between individual variables as components of $p$ yields the ``octagon domain'' \citep{mine2006theoctagon}.
Figures \ref{fig:illustration}\emph{(b,c)} illustrate PWA functions with rectangular domains.
Thus, we define the \emph{template-based piecewise affine} (TPWA) regression problem:

\begin{problem}[TPWA regression]\label{prob:PWAR}%
Given data $\{(x_k,y_k)\}_{k=1}^K$, a template $p:\Re^d\to\Re^h$ and an error bound $\epsilon> 0$, find $q$ regions $H_i\in\calH$ and affine functions $f_i(x)=A_i x + b_i$ such that \eqref{equat:error} is satisfied.
\end{problem}

Problem \ref{prob:PWAR} can be posed as a decision problem (given a bound $\qhat$, is there a TPWA function with $q\leq\qhat$ pieces?), or as an optimization problem (find a TPWA function with minimum number of pieces).
Although a solution to the decision problem can be used repeatedly to solve the optimization problem, we will focus on directly solving the optimization problem in this paper. 
Problem \ref{prob:PWAR} is closely related to the well-known problem of \emph{switched affine} (SA) regression, in which one aims to explain the data with a few affine functions, but there is no assumption on which function may explain a particular data point $(x_k, y_k)$.

\begin{problem}[SA regression]\label{prob:SAR}%
Given data $\{(x_k,y_k)\}_{k=1}^K$ and an error bound $\epsilon\geq0$, find $q$ affine functions $f_i(x) = A_i x + b_i$ such that $\forall\,k$, $\exists\,i$: $\lVert y_k-f_i(x)\rVert_\infty \leq \epsilon$.
\end{problem}

\section{Computational Complexity}\label{section:complexity}
The problem of SA regression (Problem \ref{prob:SAR}) is known to be NP-hard, even for $q=2$ \citep[\S 5.2.4]{lauer2019hybrid}.
In this section, we show that the same holds for the decision version of Problem \ref{prob:PWAR}.
We study the problem in the RAM model, wherein the problem input size is $K(d+e)+\size(p)$, where $\size(p)$ is the size needed to describe the template $p$. 

\begin{theorem}[NP-hardness]\label{thm:np-hard}%
The decision version of problem \ref{prob:PWAR} is NP-hard, even for $q=2$ and rectangular templates.
\end{theorem}

The proof reduces Problem \ref{prob:SAR} which is known to be NP-hard to Problem \ref{prob:PWAR}, and is provided in Appendix \ref{app:np-hardness-proof}.
Despite the problem being NP-hard, one can show that for fixed dimension $d$, template $p:\Re^d\to\Re^h$ and number of pieces $q$, the complexity is polynomial in the size $K$ of the data set.
Note that a similar result holds for Problem \ref{prob:SAR} \citep[Theorem 5.4]{lauer2019hybrid}.

For every $c\in\Re^h$, let $I(c)=\{k\in\Ne:1\leq k\leq K,\:x_k\in H(c)\}$ be the set of all indices $k$ such that $x_k\in H(c)$.
Also, let $\calI=\{I(c):c\in\Re^h\}$ be the set of all such index sets.


\begin{theorem}[{Polynomial complexity in $K$}]\label{thm:poly-datasize}%
For fixed dimension $d$, template $p:\Re^d\to\Re^h$ and number of pieces $q$, the complexity of Problem \ref{prob:PWAR} is bounded by $O(K^{qh})$.
\end{theorem}

Proof is provided in Appendix~\ref{app:polynomial-complexity-in-K}.
The algorithm presented in the proof of Theorem \ref{thm:poly-datasize}, although polynomial in the size of the data set, can be quite expensive in practice.
For instance, in dimension $d=2$, with rectangular regions (i.e., $h=4$) and $K=100$ data points, one would need to solve $K^h=10^8$ regression problems,%
\footnote{ In theory,  by using Sauer--Shelah's lemma (see, e.g., \citealp[Lemma~6.2.2]{harpeled2011geometric}), this number can be reduced to $\sum_{i=1}^h\binom{K}{i}\approx4\times10^6$.
This is because the \emph{VC dimension} of rectangular regions in $\Re^d$ is $2d$.}
each of which is a linear program.

In the next section, we present an algorithm for TPWA regression that is generally several orders of magnitude faster by using a \emph{top-down} approach.

\section{Top-down Algorithm for TPWA Regression}\label{section:algorithm}
We first define the concept of compatible and maximal compatible index sets.
\begin{definition}[Maximal compatible index set]\label{defin:max-comp-region}%
Consider an instance of Problem \ref{prob:PWAR}.
An index set $I\subseteq\{1,\ldots,K\}$ is \emph{compatible} if (a) $I\in\calI$ and (b) there is an affine function $f(x)=Ax+b$ such that $\forall\,k\in I$, $\lVert y_k-f(x_k)\rVert_\infty\leq\epsilon$.
A compatible index set $I$ is \emph{maximal} if there is no compatible index set $I'$ such that $I\subsetneq I'$.
\end{definition}
The key idea is that we can restrict overselves to searching over \emph{maximal} compatible index sets in order to find a solution to Problem \ref{prob:PWAR}.
See Appendix \ref{app:maximal-sufficient} for a proof.

Maximal compatible index sets can be computed by using a recursive \emph{top-down} approach (implemented in Algorithm~\ref{algo:top-down}):
Consider the lattice $\calI$ ordered by $\subseteq$ relationship. Our algorithm starts at the very top of this lattice and ``descends'' until we find maximal compatible index sets. 
At each step, we consider a current set  $I\in\calI$ (initially, $I=\{1,\ldots,K\}$) that is a candidate for being compatible and check it for compatibility.  If $I$ is not compatible, we find subsets $I_1,\ldots,I_S\subsetneq I$ using the  $\FindSubsets$ procedure, which is required to be  \emph{consistent}, as defined below.


\begin{algorithm2e}[t]
\DontPrintSemicolon
\caption{Top-down algorithm to compute maximal compatible index sets.}\label{algo:top-down}
\KwData{Data set $\{(x_k,y_k)\}_{k=1}^K$, template $p$}
\KwResult{Collection $\calS$ of all maximal compatible index sets}
$\calS\gets\emptyset$ (``compatible''); $\calU\gets\{\{1,\ldots,K\}\}$ (``to explore''); $\calV\gets\emptyset$ (``visited'')\;
\While{$\calU\setminus\calV$ is not empty}{
    Pick an index set $I$ in $\calU\setminus\calV$\;
    \eIf{$I$ is compatible}{
      Add $I$ to $\calS$; Add to $\calV$ all subsets of $I$; Remove from $\calS$ all subsets of $I$
    }%
    {
      $(I_1,\ldots,I_S)\gets\FindSubsets(I)$\tcp*[l]{satisfies Definition \ref{defn:consistency}}
      Add $I_1,\ldots,I_S$ to $\calU$; Add $I$ to $\calV$
    }
}
\Return $\calS$\;
\end{algorithm2e}

\begin{definition}[Consistency]\label{defn:consistency}
Given a non-compatible index set $I\in\calI$, a collection of index sets $I_1,\ldots,I_S\in\calI$ is said to be \emph{consistent} w.r.t.\ $I$ if (a) for each $s$, $I_s\subsetneq I$ and (b) for every compatible index set $J\subseteq I$, there is $s$ such that $J\subseteq I_s$.
\end{definition}

\begin{theorem}[Correctness of Algorithm \ref{algo:top-down}]\label{thm:algo-top-down-sound}%
If $\FindSubsets$ satisfies that for every non-compatible index set $I\in\calI$, the output of $\FindSubsets(I)$ is consistent w.r.t.\ $I$, then Algorithm \ref{algo:top-down} is correct, meaning that it terminates and the output $\calS$ is the collection of all maximal compatible index sets.
\end{theorem}
The proof is provided in Appendix~\ref{app:top-down-soundness-proof}.

\subsection{Implementation of $\FindSubsets$ using infeasibility certificates}

We now explain how to implement $\FindSubsets$ so that it is consistent.
For that, we use infeasibility certificates, which are index sets that are not compatible:

\begin{definition}[Infeasibility certificate]
An index set $C\subseteq\{1,\ldots,K\}$ is an \emph{infeasibility certificate} if there is no affine function $f(x)=Ax+b$ such that $\forall\,k\in C$, $\lVert y_k-f(x_k)\rVert_\infty\leq\epsilon$.
\end{definition}

Note that any incompatible index set $I$ contains an infeasibility certificate $C\subseteq I$ (e.g., $C=I$). However, it is quite useful to extract an infeasibility certificate $C$ that is as small as possible.
Thereafter, from an infeasibility certificate $C\subseteq I$, one can compute a consistent collection of index subsets of $I$ by tightening each component of the template \emph{independently}, in order to exclude a minimal nonzero number of indices from the infeasibility certificate, while keeping the other components unchanged.
This results in an implementation of $\FindSubsets$ that satisfies the consistency property, described in Algorithm \ref{algo:subdivision-cert}.
Figure \ref{fig:rectangular-subdivision-cert} shows an illustration for rectangular regions.
The correctness of Algorithm \ref{algo:subdivision-cert} is proved in Appendix \ref{app:correctness-findsubsets}.

\begin{algorithm2e}[t]
\DontPrintSemicolon
\caption{An implementation of $\FindSubsets$ using infeasibility certificates}\label{algo:subdivision-cert}
\KwData{Data set $\{(x_k,y_k)\}_{k=1}^K$, template $p=[p^1,\ldots,p^h]$, non-compatible index set $I=I(c)$ where $c=[c^1,\ldots,c^h]$, infeasibility certificate $C\subseteq I$}
\KwResult{A collection of index sets $I_1,\ldots,I_S$ consistent w.r.t.\ $I$}
\ForEach{$s = 1,\ldots,h$}{
  $\chat^s\gets\max\,\{p^s(x_k):k\in I,\,p^s(x_k)<\max_{\ell\in C}\,p^s(x_\ell)\}$\;
  Define $I_s=I([c^1,\ldots,c^{s-1},\chat^s,c^{s+1},\ldots,c^h])$\;
}
\Return all nonempty index sets $I_1,\ldots,I_h$\;
\end{algorithm2e}

\begin{figure}
\centering
\begin{tikzpicture}[x=2.5cm,y=1.55cm]
\draw[thick] (0,0) -- (1,0) -- (1,1) -- (0,1) -- cycle;
\node[anchor=center] at (0.5,-0.3) {$I$};
\foreach \x in {(0,0.2), (0.8,0), (1,0.8), (0.2,1), (0.7,0.5)}{
  \fill[color=black] \x circle (3pt);
}
\foreach \x in {(0.5,0.3), (0.5,0.7), (0.3,0.5), (1, 0.4)}{
  \fill[color=red] \x circle (3pt);
}
\end{tikzpicture}
\hfill
\begin{tikzpicture}[x=2.5cm,y=1.55cm]
\draw[thick,dashed] (0,0) -- (1,0) -- (1,1) -- (0,1) -- cycle;
\draw[thick] (0.5,0) -- (1,0) -- (1,1) -- (0.5,1) -- cycle;
\node[anchor=center] at (0.5,-0.3) {$I_1$};
\foreach \x in {(0.8,0), (1,0.8), (0.7,0.5)}{
  \fill[color=black] \x circle (3pt);
}
\foreach \x in {(0,0.2), (0.2,1)}{
  \draw[thick,color=black] \x circle (3pt);
}
\foreach \x in {(0.5,0.3), (0.5,0.7), (1, 0.4)}{
  \fill[color=red] \x circle (3pt);
}
\foreach \x in {(0.3,0.5)}{
  \draw[thick,color=red] \x circle (3pt);
}
\end{tikzpicture}
\hfill
\begin{tikzpicture}[x=2.5cm,y=1.55cm]
\draw[thick,dashed] (0,0) -- (1,0) -- (1,1) -- (0,1) -- cycle;
\draw[thick] (0,0) -- (0.8,0) -- (0.8,1) -- (0,1) -- cycle;
\node[anchor=center] at (0.5,-0.3) {$I_2$};
\foreach \x in {(0,0.2), (0.8,0), (0.2,1), (0.7,0.5)}{
  \fill[color=black] \x circle (3pt);
}
\foreach \x in {(1,0.8)}{
  \draw[thick,color=black] \x circle (3pt);
}
\foreach \x in {(0.5,0.3), (0.5,0.7), (0.3,0.5)}{
  \fill[color=red] \x circle (3pt);
}
\foreach \x in {(1, 0.4)}{
  \draw[thick,color=red] \x circle (3pt);
}
\end{tikzpicture}
\hfill
\begin{tikzpicture}[x=2.5cm,y=1.55cm]
\draw[thick,dashed] (0,0) -- (1,0) -- (1,1) -- (0,1) -- cycle;
\draw[thick] (0,0.4) -- (1,0.4) -- (1,1) -- (0,1) -- cycle;
\node[anchor=center] at (0.5,-0.3) {$I_3$};
\foreach \x in {(1,0.8), (0.2,1), (0.7,0.5)}{
  \fill[color=black] \x circle (3pt);
}
\foreach \x in {(0,0.2), (0.8,0)}{
  \draw[thick,color=black] \x circle (3pt);
}
\foreach \x in {(0.5,0.7), (0.3,0.5), (1, 0.4)}{
  \fill[color=red] \x circle (3pt);
}
\foreach \x in {(0.5,0.3)}{
  \draw[thick,color=red] \x circle (3pt);
}
\end{tikzpicture}
\hfill
\begin{tikzpicture}[x=2.5cm,y=1.55cm]
\draw[thick,dashed] (0,0) -- (1,0) -- (1,1) -- (0,1) -- cycle;
\draw[thick] (0,0) -- (1,0) -- (1,0.5) -- (0,0.5) -- cycle;
\node[anchor=center] at (0.5,-0.3) {$I_4$};
\foreach \x in {(0,0.2), (0.8,0), (0.7,0.5)}{
  \fill[color=black] \x circle (3pt);
}
\foreach \x in {(1,0.8), (0.2,1)}{
  \draw[thick,color=black] \x circle (3pt);
}
\foreach \x in {(0.5,0.3), (0.3,0.5), (1, 0.4)}{
  \fill[color=red] \x circle (3pt);
}
\foreach \x in {(0.5,0.7)}{
  \draw[thick,color=red] \x circle (3pt);
}
\end{tikzpicture}
 %
\vspace*{-10pt}
\caption{$\FindSubsets$ implemented by Algorithm \ref{algo:subdivision-cert} with rectangular regions.
The red dots represent the infeasibility certificate $C$. Each $I_s$ excludes at least one point from $C$ by moving one face of the box but keeping the others unchanged.}
\label{fig:rectangular-subdivision-cert}
\vspace*{-0.5cm}
\end{figure}
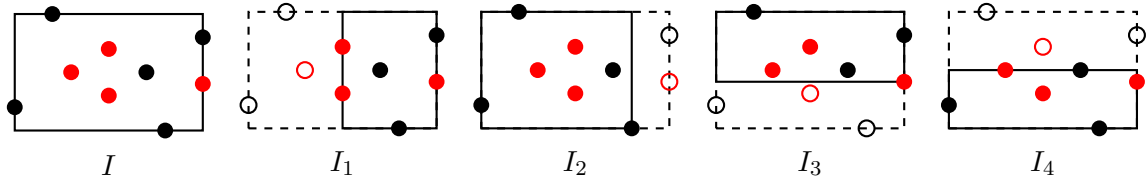


\paragraph{Good infeasibility certificates}
A trivial choice is to use $I$ as infeasibility certificate (since it is not compatible).
Although this is a valid choice, it will lead to an inefficient algorithm. To achieve efficiency,
we seek infeasibility certificates of small cardinality. 
Using the \emph{theorem of alternatives}%
\footnote{This theorem states that if a set of linear inequalities in dimension $n$ is not satisfiable, then there exists an \emph{efficiently computable} subset of $n+1$ of these inequalities that is not satisfiable \citep[Theorem 21.3]{rockafellar1970convex}.}
of Linear Programming, we can obtain a certificate $C$ that contains at most $d+2$ data points.
We also require that the points $\{x_k\}_{k\in C}$ are spatially concentrated (i.e, close to each other under some distance metric).
Indeed, concentration of the points $\{x_k\}_{k\in C}$ around some center point $\xbar$ implies that at least one set $I_1,\ldots,I_S$ produced by Algorithm \ref{algo:subdivision-cert} is small compared to the original index set $I=I(c)$, because $\xbar$ cannot be tight at all components of $H(c)$; this can be seen in Figure \ref{fig:rectangular-subdivision-cert} for rectangular regions.
This approach is described in Appendix \ref{app:findcertificates-concentrated}.

\subsection{Early stopping using set cover algorithms}

Finally, Algorithm \ref{algo:top-down} can be made much more efficient by enabling early termination if $\{1,\ldots,K\}$ is optimally covered by the compatible index sets computed so far.
For that, we add an extra step at the beginning of each iteration, that consists in (i) computing a lower bound $\beta$ on the size of an optimal cover of $\{1,\ldots,K\}$ with compatible index sets; and (ii) checking whether we can extract from $\calS$ a collection of $\beta$ index sets that form a cover of $\{1,\ldots,K\}$.
The extra step returns $\cbreak$ if (ii) is successful.
An implementation of the extra step is provided in Algorithm \ref{algo:extra-step-cover}.

\begin{algorithm2e}[t]
\DontPrintSemicolon
\caption{Extra step at the beginning of each iteration of Algorithm \ref{algo:top-down}}\label{algo:extra-step-cover}
\KwData{$\calS$, $\calU$ and $\calV$ at the iteration, $K$}
\KwResult{$\cbreak$ if we can extract from $\calS$ an optimal cover of $\{1,\ldots,K\}$ with compatible index sets; otherwise, $\ccontinue$}
Let $\alpha$ be the size of an optimal cover of $\{1,\ldots,K\}$ by index sets in $\calS$\;
Let $\beta$ be the size of an optimal cover of $\{1,\ldots,K\}$ by index sets in $\calS\cup(\calU\setminus\calV)$\;
\leIf{$\alpha\leq\beta$}{\Return $\cbreak$}{\Return $\ccontinue$}
\end{algorithm2e}

The soundness of Algorithm \ref{algo:extra-step-cover} follows from the following lemma.

\begin{lemma}\label{lem:lower-bound-cover}%
Let $\beta$ be as in Algorithm \ref{algo:extra-step-cover}.
Then, any cover of $\{1,\ldots,K\}$ with compatible index sets has size at least $\beta$.
\end{lemma}

\begin{proof}
The crux of the proof relies on the observation from the proof of Theorem~\ref{thm:algo-top-down-sound} that 
for any compatible index set $I\in\calI$, there is $J\in\calS\cup(\calU\setminus\calV)$ such that $I\subseteq J$.
It follows  that for any cover of $\{1,\ldots,K\}$ with compatible index sets, there is a cover of $\{1,\ldots,K\}$ with index sets in $\calS\cup(\calU\setminus\calV)$.
Since $\beta$ is the smallest size of such a cover, this concludes the proof of the lemma.
\end{proof}

The implementation of the extra step in Algorithm \ref{algo:top-down} is provided in Algorithm \ref{algo:top-down-early}.
The correctness of the algorithm follows from that of Algorithm \ref{algo:top-down} (Theorem \ref{thm:algo-top-down-sound}) and Algorithm \ref{algo:extra-step-cover} (Lemma \ref{lem:lower-bound-cover}).

\begin{algorithm2e}[t]
\DontPrintSemicolon
\caption{Top-down algorithm for Problem \ref{prob:PWAR}.}\label{algo:top-down-early}
[\dots]\tcp*[l]{same as in Algorithm \ref{algo:top-down}}
\While{true}{
  \lIf{Algorithm \ref{algo:extra-step-cover} outputs $\cbreak$}{
    \Return an optimal cover of $\{1,\ldots,K\}$ using index sets from $\calS$
  }
  [\dots]\tcp*[l]{same as in Algorithm \ref{algo:top-down}}
}
\end{algorithm2e}



\begin{theorem}[Optimal TPWA regression]\label{algo:algo-top-down-early-sound}%
Algorithm \ref{algo:top-down-early} solves Problem \ref{prob:PWAR} with minimal $q$.
\end{theorem}

\begin{proof}
Let $I_1,\ldots,I_q$ be the output of Algorithm \ref{algo:top-down-early}.
For each $i$, let $H_i=H(c_i)$ where $I_i=I(c_i)$ and let $f_i(x)=A_ix+b_i$ be as in (b) of Definition \ref{defin:max-comp-region}.
The fact that $H_1,\ldots,H_q$ and $f_1,\ldots,f_q$ is a solution to Problem \ref{prob:PWAR} follows from the fact that $I_1,\ldots,I_q$ is a cover of $\{1,\ldots,K\}$ and the definition of $f_1,\ldots,f_q$.
The fact that it is a solution with minimal $q$ follows from the optimality of $I_1,\ldots,I_q$ among all covers of $\{1,\ldots,K\}$ with compatible index sets.
\end{proof}

\begin{remark}
To solve the optimal set cover problems (which are NP-hard) in Algorithm \ref{algo:extra-step-cover}, we use MILP formulations.
The complexity of solving these problems grows exponentially with the size of $\calS$ and $\calS\cup(\calU\setminus\calV)$, respectively.
However, in our numerical experiments (Section \ref{sec:experiments}), we observed that the gain of stopping the algorithm early (if an optimal cover is found) systematically outbalanced the computational cost of solving the set cover problems.
\end{remark}

\section{Numerical Experiments}\label{sec:experiments}
\subsection{PWA approximation of insulin--glucose regulation model}\label{subsection:exa-glucose}
 \citet{dalla2007meal} present a nonlinear model of insulin--glucose regulation that has been 
 widely used to test artificial pancreas devices for treatment of type-1 diabetes. 
The model is nonlinear and involves $10$ state variables. However, the nonlinearity
arises mainly from the term $\Uid$ (insulin-dependent glucose utilization) involving two state variables, say $x_1$ and $x_2$ (namely, the level of insulin in the interstitial fluid, and the glucose mass in rapidly equilibrating tissue):
\[
\Uid(x_1,x_2) = \frac{(3.2667 + 0.0313x_1)x_2}{253.52 + x_2}.
\]
We consider the problem of approximating $\Uid$ with a PWA model, thus converting the entire model into a PWA model. 
Therefore, we simulated trajectories and collected $K=100$ values of $x_1$, $x_2$ and $\Uid(x_1,x_2)$; see Figure \ref{fig:dallaman}\emph{(a)}.
For three different values of the error tolerance, $\epsilon\in\{0.2,0.1,0.05\}$, we used Algorithm \ref{algo:top-down-early} to compute a PWA regression of the data with rectangular domains.
The results of the computations are shown in Figure \ref{fig:dallaman}\emph{(b,c,d)}.
The computation times are respectively $1$, $22$ and $112$ secs%
\footnote{On a laptop with Intel Core i7-7600u and 16 GB RAM running Windows, using Gurobi\textsuperscript{TM} as (MI)LP solver.}.
Finally, we evaluate the accuracy of the PWA regression for the modeling of the glucose-insulin evolution by simulating the system with $\Uid$ replaced by the PWA models.
The results are shown in Figure \ref{fig:dallaman}\emph{(e,f)}.
We see that the PWA model with $\epsilon=0.05$ induces a prediction error less than $2\,\%$ over the whole simulation interval, which is a significant improvement compared to the PWA models with only $1$ affine piece ($\epsilon=0.2$) or $2$ affine pieces ($\epsilon=0.1$).

Finally, we compare with switched affine regression and classical PWA regression.
To find a switched affine model, we solved Problem \ref{prob:SAR} with $\epsilon=0.05$ and $q=3$ using a MILP approach.
The computation is very fast ($<0.5$ secs); however, the computed clusters of data points (see Figure \ref{fig:dallaman-milp} in Appendix \ref{app:supplementary-material}) do not allow to learn a PWA model, thereby hindering the derivation of a model for $\Uid$ that can be used for simulation and analysis.

\begin{figure}[t]
\centering
\subfigure[$\Uid$]{\includegraphics[width=0.245\textwidth]{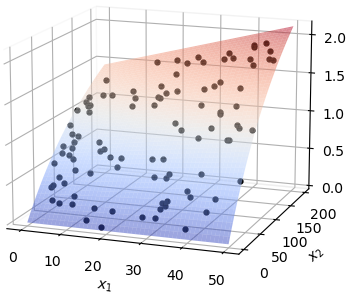}}
\hfill
\subfigure[$\epsilon=0.2$]{\includegraphics[width=0.245\textwidth]{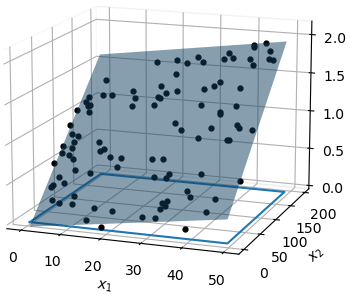}}
\hfill
\subfigure[$\epsilon=0.1$]{\includegraphics[width=0.245\textwidth]{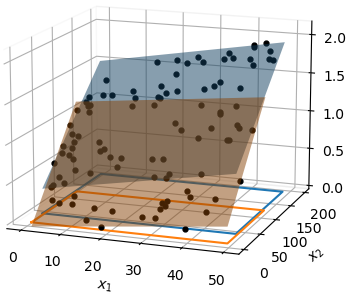}}
\hfill
\subfigure[$\epsilon=0.05$]{\includegraphics[width=0.245\textwidth]{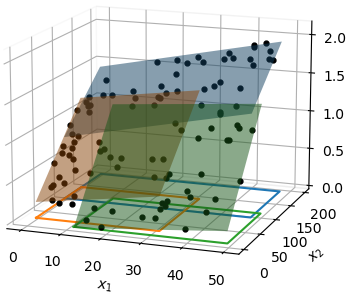}}
\subfigure[Simulated trajectories]{\includegraphics[width=0.45\textwidth]{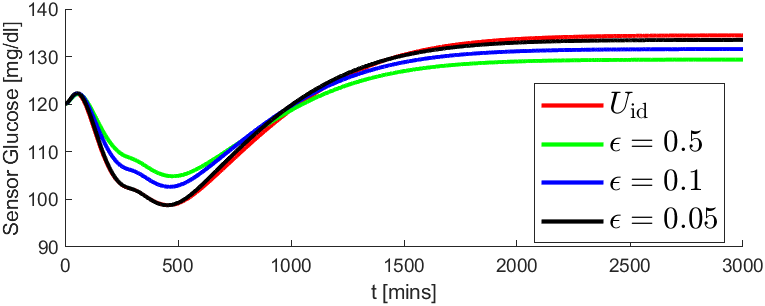}}
\hfill
\subfigure[Average error over 50 simulations]{\includegraphics[width=0.45\textwidth]{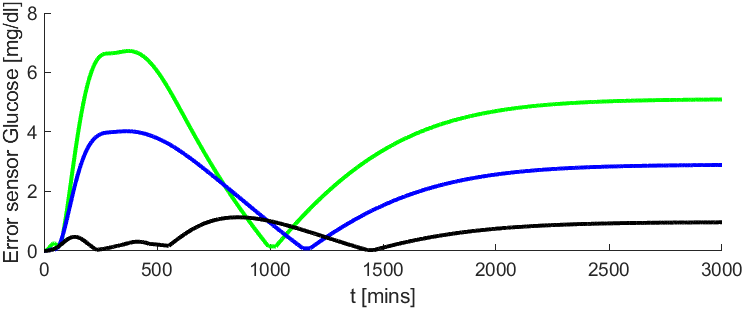}}
\caption{Glucose--insulin system.
\emph{(a)}:  $100$ sampled points (black dots) on the graph of $\Uid$ (surface).
\emph{(b)}, \emph{(c)}, \emph{(d)}: Optimal TPWA regression for various error tolerances $\epsilon$.
\emph{(e)}: Simulations using the nonlinear model versus the PWA approximations.
\emph{(f)}: Error between nonlinear and PWA models averaged over $50$ simulations with different initial conditions.}
\label{fig:dallaman}
\vspace*{-0.5cm}
\end{figure}

\subsection{Hybrid system identification: double pendulum with soft contacts}\label{subsection:exa-pendulum}

\begin{figure}[t]
\centering
\subfigure[Schematic]{%
\begin{tikzpicture}[scale=1.5]
\pgfmathsetmacro{\rad}{20}
\pgfmathsetmacro{\sad}{-10}
\coordinate (m1) at (0, 0.3);
\draw[draw=black,thick] (-1, 0) -- (1, 0);
\draw[draw=black,thick,fill,black,pattern=north east lines] (-0.5, 0) -- (0.5, 0) -- (m1) -- cycle;
\filldraw[fill=gray!30,draw=white] (m1) -- ([shift=(150:0.8cm)] m1) -- ([shift=(30:0.8cm)] m1) -- cycle;
\draw[dashed] (m1) -- ([shift=(90:0.9cm)] m1);
\draw ([shift=(90:0.45cm)] m1) -- node[pos=0,right] {$\theta_1$} ([shift=(90+\rad:0.45cm)] m1);
\coordinate (m2) at ([shift=(90+\rad:1cm)] m1);
\filldraw[fill=gray!30,draw=white] (m2) -- ([shift=(130+\rad:0.7cm)] m2) -- ([shift=(50+\rad:0.7cm)] m2) -- cycle;
\draw[dashed] (m2) -- ([shift=(90+\rad:0.9cm)] m2);
\draw ([shift=(90+\rad:0.45cm)] m2) -- node[pos=0,left] {$\theta_2$} ([shift=(90+\sad:0.45cm)] m2);
\coordinate (m3) at ([shift=(90+\sad:0.9cm)] m2);
\draw[line width=1pt] (m1) -- (m2);
\draw[line width=1pt] (m2)  -- (m3);
\draw[fill=black] (m2) circle (0.09);
\draw[fill=black] (m3) circle (0.09);
\draw[dashed] (m1) -- ([shift=(150:0.79cm)] m1);
\draw[dashed] (m1) -- ([shift=(30:0.79cm)] m1);
\draw[dashed] (m2) -- ([shift=(130+\rad:0.69cm)] m2);
\draw[dashed] (m2) -- ([shift=(50+\rad:0.69cm)] m2);
\end{tikzpicture}%
}
\hfill
\subfigure[TPWA regression $\epsilon=0.01$]{\includegraphics[width=0.35\textwidth]{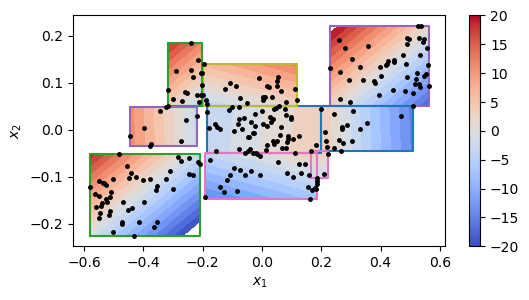}}
\hfill
\subfigure[Computation times]{\includegraphics[width=0.3\textwidth]{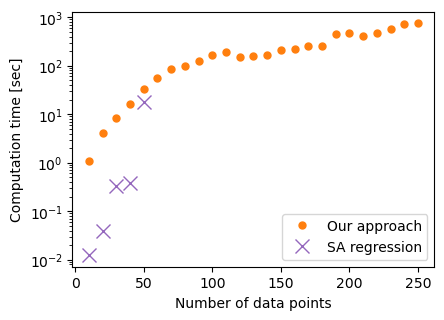}}
\caption{Inverted double pendulum with soft contacts.
\emph{(a)}: Elastic contact forces apply when $\theta$ is outside gray region,
\emph{(b)}: Optimal TPWA regression of the data with rectangular domains.
\emph{(c)}: Comparison with MILP approach for SA regression.
Time limit is set to $1000$ secs.}
\label{fig:pendulum}
\vspace*{-0.5cm}
\end{figure}
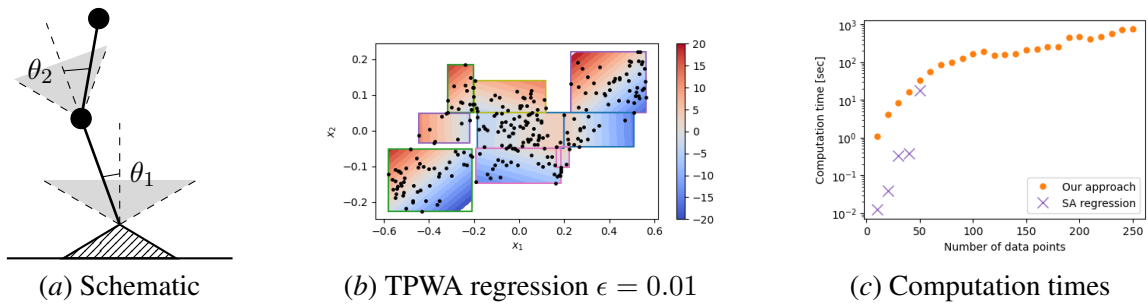

We consider a hybrid linear system consisting in an inverted double pendulum with soft contacts at the joints, as depicted in Figure \ref{fig:pendulum}\emph{(a)}.
This system has nine linear modes, depending on whether the contact force of each joint is inactive, active on the left or active on the right (see \citealp{aydinoglu2020contactaware}).
Our goal is to learn these linear modes as well as their domain of validity, from data.
For that, we simulated trajectories and collected $K=250$ sampled values of $\theta_1$, $\theta_2$ and the force applied on the lower joint.
We used Algorithm \ref{algo:top-down-early} to compute a PWA regression of the data with rectangular domains and with error tolerance $\epsilon=0.01$.
The result is shown in Figure \ref{fig:pendulum}\emph{(b)}.
The number of iterations of the algorithm was about $23000$ for a total time of $800$ secs.

We see that the affine pieces roughly divide the state space into a grid of $3\times 3$ regions.
This is consistent with our ground truth model, in which the contact force at each joint has three linear modes depending only on the angle made at the joint.
The PWA regression provided by Algorithm \ref{algo:top-down-early} allows us to learn this feature of the system from data, without assuming anything about the system except that the domains of the affine pieces are rectangular.

Finally, we compare with switched affine (SA) regression and classical PWA regression.
The MILP approach to solve the SA regression (Problem \ref{prob:SAR}) with $\epsilon=0.01$ and $q=9$ could not handle more than $51$ data points within reasonable time ($1000$ secs); see Figure \ref{fig:pendulum}\emph{(c)}.
Furthermore, the computed clusters of data points (see Figure \ref{fig:pendulum-milp} in Appendix \ref{app:supplementary-material}) do not allow to learn a PWA model, thereby hindering to learn important features of the system.

\section*{Conclusion}

We have introduced the template-based piecewise affine regression problem,
analyzed its computational complexity and provided a top-down algorithm  based on infeasibility certificates.
Numerical examples show that the algorithm compares favorably to state-of-the-art approaches for PWA regression.
In future work, we plan to study extensions to a larger class of shapes for the domains while investigating connections to  approximation algorithms for geometric set cover problems.

\acks{ This research was funded in part by the Federation Wallonie--Bruxelles (WBI) and the US National Science Foundatton (NSF) under award numbers 1836900 and 1932189.}

\bibliography{myrefs.bib}

\clearpage
\appendix

\section{Proof of NP-Hardness}\label{app:np-hardness-proof}
Given two vectors or matrices $u$ and $v$, their horizontal (resp.\ vertical) concatenation is denoted by $[u,v]$ (resp.\ $[u;v]$).
For positive integers $d$ and $e$ and a scalar $\alpha$, we denote by $[\alpha]_d$ (resp.\ $[\alpha]_{e,d}$) the vector in $\Re^d$ (resp.\ matrix in $\Re^{e\times d}$) whose components are all equal to $\alpha$.

\bigskip

\begin{proofwithname}[Proof of Theorem \ref{thm:np-hard}:]
For the simplicity of notation, we will restrict here to piecewise \emph{linear} models (i.e., with $f_i(x)=A_ix$) since PWA models can be obtained from linear ones by augmenting each data point $x_k$ with a component equal to $1$, i.e., $x_k\gets[x_k;[1]_1]\in\Re^{d+1}$.

We will reduce Problem \ref{prob:SAR} to Problem \ref{prob:PWAR}.
Therefore, consider an instance of Problem \ref{prob:SAR} consisting in a data set $\calD=\{(x_k,y_k)\}_{k=1}^K\subseteq\Re^d\times\Re^e$ and tolerance $\epsilon$.
From $\calD$, we build another data set $\calD'\subseteq\Re^{d+K}\times\Re^e$ with $\lvert\calD'\rvert=4K$ as follows.
For each $1\leq k\leq K$, we let $\chi_k\in\Re^K$ be the indicator vector of the $k$\textsuperscript{th} component.
We define
\[
\calD'=\bigcup_{\sigma\in\{-1,1\}}\bigcup_{k=1}^K\big[\{([\sigma x_k;\chi_k],\sigma y_k),([[0]_d;\chi_k],[\sigma\epsilon]_e)\}\big].
\]
Also, we let $p$ be the rectangular template in $\Re^{d+K}$, which is linear with $\size(p)=2(d+K)^2$.

\emph{Main step:} We show that Problem \ref{prob:SAR} with $\calD$, $\epsilon$ and $q=2$ has a solution iff Problem \ref{prob:PWAR} with $\calD'$, $p$, $\epsilon$ and $q=2$ has a solution.

\emph{Proof of ``if direction''.}
Assume that Problem \ref{prob:PWAR} has a solution given by $H_1,H_2\subseteq\Re^{d+K}$ and $A_1,A_2\in\Re^{e\times(d+K)}$, and for each $i$, decompose $A_i=[B_i,C_i]$, wherein $B_i\in\Re^{e\times d}$ and $C_i\in\Re^{e\times K}$.
We will show that $B_1,B_2$ provide a solution to Problem \ref{prob:SAR}.

Therefore, fix $1\leq k\leq K$.
Using the pigeon-hole principle, let $i\in\{1,2\}$ be such that at least two points in $\{[x_k;\chi_k],[-x_k;\chi_k],[[0]_d;\chi_k]\}$ belong to $H_i$.
Then, by the convexity of $H_i$, it holds that $[[0]_d;\chi_k]\in H_i$.
For definiteness, assume that $[x_k;\chi_k]\in H_i$.
Since $H_1,H_2$ and $A_1,A_2$ provide a solution to Problem \ref{prob:PWAR}, it follows that
\[
\lVert y_k - B_ix_k - C_i\chi_k \rVert_\infty \leq \epsilon, \quad \lVert [\epsilon]_e - C_i\chi_k \rVert_\infty \leq \epsilon, \quad \lVert [-\epsilon]_e - C_i\chi_k \rVert_\infty \leq \epsilon.
\]
The last two conditions imply that $C_i\chi_k=0$, so that $\lVert y_k - B_ix_k \rVert_\infty \leq \epsilon$.
Since $k$ was arbitrary, this shows that $B_1,B_2$ provide a solution to Problem \ref{prob:SAR}; thereby proving the ``if direction''.

\emph{Proof of ``only if direction''.}
Assume that Problem \ref{prob:SAR} has a solution given by $A_1,A_2\in\Re^{e\times d}$.
For each $1\leq k\leq K$, define the intervals $I_{1,k},I_{2,k}\subseteq\Re$ as follows: $I_{i,k}=[0,1]$ if $\lVert y_k-A_ix_k\rVert_\infty\leq\epsilon$, and $I_{i,k}=\{0\}$ otherwise.
Now, define the rectangular regions $H_1,H_2\subseteq\Re^{d+K}$ as follows: $H_i=\Re^d\times I_{i,1}\times\cdots\times I_{i,K}$.
Also define the matrices $B_1,B_2\in\Re^{e\times(d+K)}$ as follows: $B_i=[A_i,[0]_{e,K}]$.
We will show that $H_1,H_2$ and $B_1,B_2$ provide a solution to Problem \ref{prob:PWAR}.

Therefore, fix $1\leq k\leq K$ and $i\in\{1,2\}$.
First, assume $\lVert y_k-A_ix_k\rVert_\infty\leq\epsilon$.
We show that (a) $[x_k;\chi_k]$, $[-x_k;\chi_k]$ and $[[0]_d;\chi_k]$ belong to $H_i$, and (b)
\[
\lVert y_k - B_i[x_k;\chi_k] \rVert_\infty \leq \epsilon, \quad \lVert -y_k - B_i[-x_k;\chi_k] \rVert_\infty \leq \epsilon, \quad \lVert [\pm\epsilon]_e - B_i[[0]_d;\chi_k] \rVert_\infty \leq \epsilon.
\]
This is direct (a) by the definition of $I_{i,k}$, and (b) by the definition of $B_i$.
Now, assume that $\lVert y_k-A_ix_k\rVert_\infty\leq\epsilon$ does not hold.
We show that $[x_k;\chi_k]$, $[-x_k;\chi_k]$ do not belong to $H_i$.
This is direct since $1\notin I_{i,k}$.
Thus, we have shown that $H_1,H_2$ and $B_1,B_2$ provide a solution to Problem \ref{prob:PWAR}; thereby proving the ``only if direction''.

Hence, we have built a polynomial reduction from Problem \ref{prob:SAR} to Problem \ref{prob:PWAR}.
Since Problem \ref{prob:SAR} is NP-hard \citep[\S 5.2.4]{lauer2019hybrid}, this shows that Problem \ref{prob:PWAR} is NP-hard as well.
\end{proofwithname}

\begin{remark}
The reduction from Problem \ref{prob:SAR} to Problem \ref{prob:PWAR} in the above proof relies on the fact that $q=2$.

First, the fact that Problem \ref{prob:PWAR} is NP-hard with $q=2$ implies that Problem \ref{prob:PWAR} is NP-hard with any $q\geq 2$.
Indeed, if Problem \ref{prob:PWAR} can be solved in polynomial time for some $q=\qhat>2$, then one can add spurious data points (e.g., at a far distance of the original data points) to enforce the value of $\qhat-2$ affine pieces of the PWA function.
The satisfiability of Problem \ref{prob:PWAR} with $q=\qhat$ and the augmented data set is then equivalent to the satisfiability of Problem \ref{prob:PWAR} with $q=2$ and the original data set.

Second, given $\qhat\geq2$ and any template $p$, a construction similar to the one used in the above proof can be used to reduce Problem \ref{prob:SAR} to Problem \ref{prob:PWAR} at the cost of introducing a small gap in the reduction.
Indeed, fix $\lambda>0$ and consider the data set $\calD'=\bigcup_{t=1}^{q+1}\{([x_k;t\lambda\chi_k],y_k)\}_{k=1}^K$.
Then, one can show that if Problem \ref{prob:PWAR} with $\calD'$, $p$, $\epsilon=\epsilonhat(1-\frac2\lambda)$ and $q=\qhat$ has a solution, then Problem \ref{prob:SAR} with $\calD$, $\epsilon=\epsilonhat$ and $q=\qhat$ has a solution.
The gap corresponds to the factor $1-\frac2\lambda$, which can be made arbitrarily close to one.
\end{remark}

\section{Proof of Polynomial Complexity Bound}\label{app:polynomial-complexity-in-K}
\begin{proofwithname}[Proof of Theorem \ref{thm:poly-datasize}:]
The crux of the proof is to realize that $\lvert\calI\rvert\leq K^h+1$.

For every $c\in\Re^h$, define $P(c)=\{p(x_k) : 1\leq k\leq K,\:p(x_k)\leq c\}$ and let $\calP=\{P(c):c\in\Re^h\}$.
It holds that $\lvert\calP\rvert\leq K^h+1$.
Furthermore, there is a one-to-one correspondence between $\calP$ and $\calI$ given by: $P(c)\mapsto I(c)$.
Indeed, it is clear that if $I(c_1)=I(c_2)$, then $P(c_1)=P(c_2)$.
On the other hand, if $I(c_1)\nsubseteq I(c_2)$, then there is at least one $k$ such that $p(x_k)\leq c_1$ but $p(x_k)\nleq c_2$.
This implies that $P(c_1)\nsubseteq P(c_2)$.
Therefore, $\lvert\calP\rvert = \lvert\calI\rvert \in O(K^h)$.

Now, Problem \ref{prob:PWAR} can be solved by enumerating the $L=K^h$ nonempty index sets $I_1,\ldots,I_L$ in $\calI$, and keeping only those $I_\ell$ for which we can fit an affine function over the data $\{(x_k,y_k)\}_{k\in I_\ell}$ with error bound $\epsilon$.
Next, we enumerate all combinations of $q$ such index sets that cover the indices $\{1,\ldots,K\}$.
There are at most $L^q$ such combinations.
This concludes the proof of the theorem.
\end{proofwithname}

\section{Maximal Compatible Index Sets}\label{app:maximal-sufficient}
\begin{lemma}
Let $q$ be given.
Problem \ref{prob:PWAR} has a solution iff it has a solution wherein the regions correspond to maximal compatible index sets.
\end{lemma}

\begin{proof}
The ``if direction'' is clear.
We prove the ``only if direction''.
Consider a solution of Problem \ref{prob:PWAR} with regions $H_1,\ldots,H_q$.
For each $1\leq i\leq q$, there is a maximal compatible index set $I_i=I(c_i)$ such that $H_i\cap\{x_k\}_{k=1}^K\subseteq H(c_i)$.
Since $\{x_k\}_{k=1}^K\subseteq\bigcup_{i=1}^qH_i$, it holds that $\{x_k\}_{k=1}^K\subseteq\bigcup_{i=1}^qH(c_i)$.
Hence, $H(c_1),\ldots,H(c_q)$, along with affine functions $f_i(x)=A_ix+b_i$ satisfying (b) in Definition \ref{defin:max-comp-region}, provide a solution to Problem \ref{prob:PWAR}, concluding the proof.
\end{proof}

\section{Correctness of Top-Down Algorithm}\label{app:top-down-soundness-proof}
\begin{proofwithname}[Proof of Theorem~\ref{thm:algo-top-down-sound}]
Termination follows from the fact that each index set $I\in\calI$ is picked at most once, because when some $I\in\calI$ is picked, it is then added to the collection $\calV$ of visited index sets, so that it cannot be picked a second time.
Since $\calI$ is finite, this implies that the algorithm terminates in a finite number of steps.

Now, we prove that, upon termination, any maximal compatible index set is in the output $\calS$ of the algorithm.
Therefore, let $J$ be a maximal compatible index set.
Then, among all sets $I$ picked during the execution of the algorithm and satisfying $J\subseteq I$, let $I^*$ have minimal cardinality.
Such an index set exists since $J\subseteq\{1,\ldots,K\}$.
We will show that:

\emph{Main result.}
$I^*=J$.

\emph{Proof of main result.}
For a proof by contradiction, assume that $I^*\neq J$.
Since $J$ is maximal and $J\subsetneq I^*$, $I^*$ is not compatible.
Hence, the index sets $(I_1,\ldots,I_S)=\FindSubsets(I^*)$ were added to $\calU$.
Using the assumption on $\FindSubsets$, let $s$ be such that $J\subseteq I_s\subsetneq I^*$.
Since $I_s$ must have been picked during the execution of the algorithm, this contradicts the minimality of the cardinality of $I^*$, concluding the proof of the main result.

Thus, $J$ was picked during the execution of the algorithm.
Since it is compatible, it was added to $\calS$ at the iteration at which it was picked, and since it is maximal, it is not removed at later iterations.
Hence, upon termination, $J\in\calS$.
Since $J$ was arbitrary, this concludes the proof that, upon termination, $\calS$ contains all maximal compatible index sets.

Finally, we show that, upon termination, $\calS$ contains only maximal compatible index sets.
This follows from the fact that, at each iteration of the algorithm, for any distinct $I_1,I_2\in\calS$, it holds that $I_1\nsubseteq I_2$ and $I_2\nsubseteq I_1$.
Indeed, when $I_1$ is added to $\calS$, all subsets of $I_1$ are removed from $\calS$ and are added to $\calV$ so that they are not picked at later iterations.
The same holds for $I_2$.
This concludes the proof of the theorem.
\end{proofwithname}

\section{Correctness of \FindSubsets}\label{app:correctness-findsubsets}
\begin{lemma}\label{lem:certificate-compatible}%
If $C$ is an infeasibility certificate, then every $I\subseteq\{1,\ldots,K\}$ satisfying $C\subseteq I$ is not compatible.
\end{lemma}

\begin{proof}
Straightforward from (b) in Definition \ref{defin:max-comp-region}.
\end{proof}

\begin{theorem}[Correctness of Algorithm \ref{algo:subdivision-cert}]\label{lem:subdivision-cert-sound}%
For every non-compatible index set $I\in\calI$, the output $I_1,\ldots,I_S$ of Algorithm \ref{algo:subdivision-cert} is consistent w.r.t.\ $I$.
\end{theorem}

\begin{proof}
Let $J\subseteq I$ be compatible.
Using that $C\nsubseteq J$ (Lemma \ref{lem:certificate-compatible}), let $s$ be a component such that $\max_{k\in J}\,p^s(x_k)<\max_{k\in C}\,p^s(x_k)$.
It holds that $J\subseteq I_s$.
Since $J$ was arbitrary, this concludes the proof.
\end{proof}

\section{Spatially Concentrated Infeasibility Certificates}\label{app:findcertificates-concentrated}
\subsection{Optimization program formulation}

Given a center point $\xbar$ and a non-compatible index set $I\subseteq\{1,\ldots,K\}$, we consider the following Linear Program: with variables $\lambda_k\in\Re$, $\forall\,k\in I$,
\begin{equation}\label{equat:dual-lp}
\begin{array}{@{}rl@{}}
\text{minimize} & \sum_{k\in I}\lvert\lambda_k\rvert\lVert x_k-\xbar\rVert^2 \\
\text{s.t.} & \sum_{k\in I} \lambda_k[x_k;[1]_1] = [0]_{d+1} \;\wedge\; \sum_{k\in I} \lambda_ky_k \ngeq -\sum_{k\in I} \lvert\lambda_k\rvert\epsilon. \\
\end{array}
\end{equation}
From the theorem of alternatives of Linear Programming, it holds that \eqref{equat:dual-lp} has a feasible solution $\{\lambda_k\}_{k\in I}$ satisfying that at most $d+2$ variables are nonzero.
The objective function of \eqref{equat:dual-lp} tends to put zero value to $\lambda_k$ whenever $\lVert x_k-\xbar\rVert_\infty$ is large.
This promotes proximity of the point $x_k$ to $\xbar$ when $\lambda_k\neq0$.%
\footnote{Note that $L^1$ regularization costs are often used in machine learning to induce sparsity of the optimal solution \citep[p.\ 304]{boyd2004convex}.
Here, we use a \emph{weighted} $L^1$ regularization cost to induce a sparsity pattern dictated by the geometry of the problem.}
In our experiments, we used $\xbar=\frac1{\lvert I\rvert}\sum_{k\in I}x_k$.

\subsection{Complexity analysis under strong assumptions}\label{ssec:complexity-analysis}

Consider the domain $D=[0,1]^d$ and let $f:D\to\Re^e$ by a PWA function with $q$ pieces, whose domains $H_1,\ldots,H_q\subseteq D$ are rectangles (e.g., the black rectangles in Figure \ref{fig:certificate-rectangle}).
Let $N\in\Ne_{>0}$ and consider the sampled input set $\calX\subseteq D$ obtained by griding uniformly $D$ with $N$ points along each axis (hence, $\lvert\calX\rvert=N^d$).
Now consider the data set $\calD=\{(x,f(x)):x\in\calX\}$.
We aim to solve Problem \ref{prob:PWAR} with data set $\calD$, $\epsilon=0$, $q$ as above and the rectangular template.
We will compare the efficiency of the naïve approach (outlined in the proof of Theorem \ref{thm:poly-datasize}) with the top-down approach presented in Algorithm \ref{algo:top-down}.
In particular, we will investigate the case $N\to\infty$.

\paragraph{Naïve approach}

The naïve approach consists in enumerating all subsets of $\calX$ that are compatible with the rectangular template.
There are $\big(\frac{N(N-1)}2\big)^d$ such subsets (choose a lower bound and an upper bound along each axis).
This gives a lower bound on the computational complexity of the naïve approach, that grows polynomially with $N$.

\paragraph{Top-down approach}

We let $\FindSubsets$ be implemented as in Algorithm \ref{algo:subdivision-cert} and we assume that, at each call, the associated infeasibility certificate $\calX_\infeas$ consists in $d+2$ points concentrated around the center $\xbar$ of $H(c)\cap\calX$.
See Figure \ref{fig:certificate-rectangle} for an illustration, where $H(c)$ is the blue rectangle, $\xbar$ is the blue dot and $\calX_\infeas$ is the red dots.
This assumption holds naturally when $H(c)\cap\calX$ contains a lot of points (which is the case when $N$ is large), and the certificate is computed using \eqref{equat:dual-lp}.
It follows that the parameters $\chat^1,\ldots,\chat^h$ computed by $\FindSubsets$ satisfy that for all $1\leq s\leq h$, $p^s(\xbar)\approx\chat^s$ or $\chat^s\approx c_i^s$, wherein $i\in\{1,\ldots,q\}$ is such that $\xbar\in H_i$ and $c_i$ is such that $H_i=H(c_i)$.
Hence, the rectangles $H(c_1),\ldots,H(c_h)$ computed by $\FindSubsets$ satisfy that for all $1\leq s\leq h$, either the volume of $H(c_s)$ is half of that of $H(c)$ (since one face is tight at $\xbar$, the center of $H(c)$) or the number of components $c_s^t$ of $c_s$ satisfying $c_s^t\leq c_i^t$, wherein $H_i$ contains the center of $H(c_s)$, is strictly larger than that of $c$.
By adding the natural assumption that all regions $H_i$ have a volume of at least $\nu\in(0,1]$ (user-provided) and discarding regions with volume smaller than $\nu$, we get that the algorithm cannot divide the volume of a region more than $-\log_2(\nu)$.
Hence, the depth of the tree underlying the algorithm is upper bounded by $h-\log_2(\nu)$.
Since, each node of the tree has at most $h$ children (the subsets given by $\FindSubsets$), the number of rectangles encountered during the algorithm is upper bounded by $h^{h-\log_2(\nu)}$.
Note that this upper bound on the complexity of the algorithm is independent of $N$.

\begin{figure}
\centering
\begin{tikzpicture}
\begin{axis}[
	width=0.5\textwidth,height=0.5\textwidth,
	xmin=-0.1,xmax=1.1,
	ymin=-0.1,ymax=1.1
]
\draw[black] (0,0) rectangle (0.7,0.5);
\draw[black] (0.7,0) rectangle (1,0.7);
\draw[black] (0,0.5) rectangle (0.4,1);
\draw[black] (0.4,0.5) rectangle (0.7,1);
\draw[black] (0.7,0.7) rectangle (1,1);
\draw[blue] (0.1,0.15) rectangle (0.6,0.55);
\fill[blue] (0.35,0.35) circle (2pt);
\fill[red] (0.35,0.365) circle (2pt);
\fill[red] (0.35,0.5) circle (2pt);
\end{axis}
\end{tikzpicture}
\caption{Illustration of the top-down approach in Subsection \ref{ssec:complexity-analysis}.}
\label{fig:certificate-rectangle}
\end{figure}
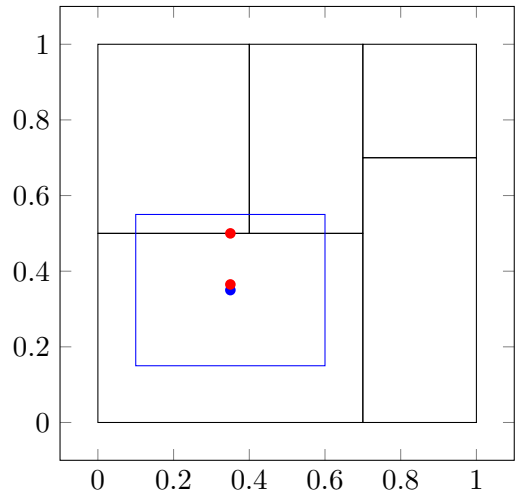

\section{Supplementary material}\label{app:supplementary-material}

\begin{figure}[h]
\centering
\includegraphics[height=6cm]{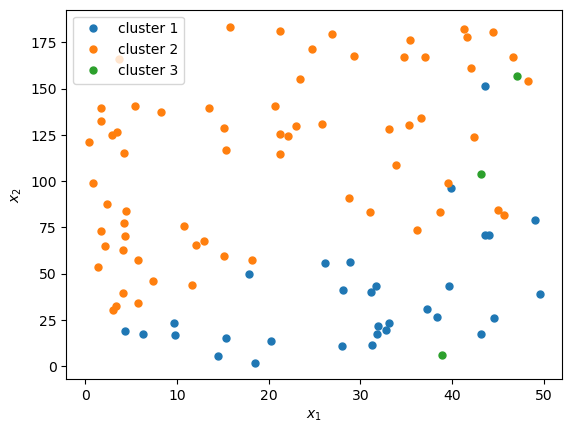}
\caption{Clusters of data points from SA regression of the data set from the insulin--glucose regulation system in Subsection \ref{subsection:exa-glucose}.}
\label{fig:dallaman-milp}
\end{figure}

\begin{figure}[h]
\centering
\includegraphics[height=6cm]{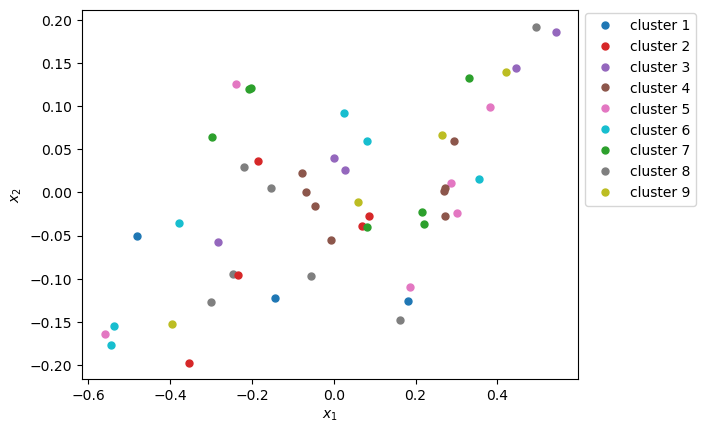}
\caption{Clusters of data points from SA regression of $51$ data points from the inverted double pendulum with soft contacts in Subsection \ref{subsection:exa-pendulum}.}
\label{fig:pendulum-milp}
\end{figure}

\end{document}